\documentclass{article}
\usepackage[margin=1in]{geometry}  
\usepackage{amsmath, amssymb, amsthm}
\usepackage{caption}
\usepackage{graphicx}
\usepackage{hyperref}

\newtheorem* {lemma}{Lemma}
\newtheorem* {theorem}{Theorem}

\theoremstyle{plain}
\newtheorem*{corollary}{Corollary}

\title{An Alternative Proof of Dilworth's Theorem via Local Chain Merges}
\author{Tao Zhang \\
\href{mailto:14126186@bjtu.edu.cn}{14126186@bjtu.edu.cn}}
\date{}  
\begin{document}
\maketitle

\begin{abstract}
Several elegant proofs of Dilworth's theorem for finite posets already exist in the literature. In this note, we present an alternative inductive proof using a local ``legal merge'' lemma on chain ends. By combining this lemma with the bound established by Dilworth's theorem, we observe a direct consequence: every chain decomposition can be transformed into a minimal chain decomposition through a finite sequence of such legal merges.
\end{abstract}

\subsection*{Coloring and Legal Merges}
Let $S$ be a finite poset with an initial $k$-chain decomposition $\mathcal{C}=\{C_1, \dots, C_k\}$. We define a \textit{coloring scheme} on $S$ by assigning color $i$ to every element $x \in C_i$. Thus, elements share a color if and only if they belong to the same chain in $\mathcal{C}$. 

For an arbitrary chain decomposition $\mathcal{C}'$ of $S$, a \textit{legal merge} is an operation defined as follows: if two chains in $\mathcal{C}'$ have minimal elements of the same color, we merge their monochromatic bottom segments. Since all elements of a given color belong to the same chain in $\mathcal{C}$ and are thus pairwise comparable, the union of these segments forms a linear order. Specifically, the segment whose own maximal element is smaller is merged into the chain containing the segment with the larger maximal element, ensuring the merged set remains a valid chain.

\begin{lemma}
For any chain decomposition $\mathcal{C}'$ of $S$, there exists a finite sequence of legal merges starting from $\mathcal{C}'$ and resulting in a decomposition $\mathcal{C}''$ where the minimal elements of the chains have pairwise distinct colors.
\end{lemma}

\begin{proof}
Let $\Phi$ be the total number of monochromatic segments in $\mathcal{C}'$. Since $S$ is finite, $\Phi$ must be finite. Each legal merge reduces $\Phi$ by exactly $1$. Since $\Phi \ge k$, $\Phi$ cannot decrease indefinitely. Therefore, the process must terminate.
\end{proof}

\begin{figure}[htbp]
    \centering
    \includegraphics[width=0.7\textwidth]{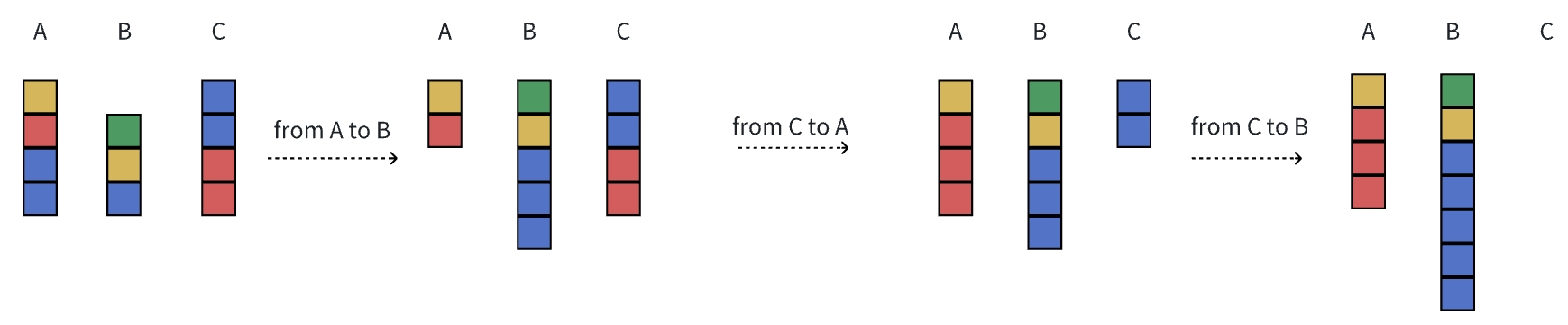}
    \vspace{5pt}
    \captionsetup{font=footnotesize, labelfont=bf} 
    \caption{An example of the legal merge.}
\end{figure}

\begin{theorem}[Dilworth]
The minimum number of chains needed to cover a finite poset $S$ is equal to the maximum size of an antichain in $S$.
\end{theorem}

\begin{proof}
Clearly, the minimum number of chains is at least the maximum size of an antichain. It remains to show that equality holds, which we prove by induction on $|S|$.

The base case $|S|=1$ is trivial. Let $e$ be a maximal element of $S$, and let $S'=S\setminus\{e\}$. Let $k$ be the size of the maximum antichain in $S'$. By the inductive hypothesis, $S'$ admits a $k$-chain decomposition $\{C_1, \dots, C_k\}$. Assign to each chain $C_i$ a distinct color $i$, and color all elements within $C_i$ with this color. Define 
\[ L_i = \{c \in C_i \mid c \preceq e\}, \quad U_i = \{c \in C_i \mid c \npreceq e\}. \]

Let $U = \bigcup_{i=1}^k U_i$. Since $e$ is a maximal element of $S$, $e$ is incomparable to all elements in $U$. By induction, $U$ decomposes into $r'$ chains $\mathcal{U}'$, where $r'$ is the size of the largest antichain in $U$. By the Lemma, $\mathcal{U}'$ can be adjusted to a decomposition $\mathcal{U}''$ with $r''$ chains of distinct minimal element colors, where $r'' \le r' \le k$.

Each chain in $\mathcal{U}''$ pairs with exactly one $L_i$ by the color of its minimal element, with the paired two concatenated end-to-end to form a single chain. If $r'' < k$, we append $e$ to any unpaired $L_i$ to form a chain, yielding a partition of $S$ into $k$ chains. If $r'' = k$, let the element $e$ form a separate chain; however, since $e$ is incomparable to all elements in $U$, any maximum antichain of $U$ together with $e$ forms an antichain of size $k+1$. The equality remains satisfied.
\end{proof}
\medskip
Combining the merging lemma with Dilworth’s theorem, we immediately obtain the following structural corollary.
\begin{corollary}
Every chain decomposition can be transformed into a minimal decomposition after finitely many legal merges, where the final number of chains equals the width of the poset.
\end{corollary}

\begin{proof}
Let $w$ be the width of the poset $S$. By Dilworth's theorem, there exists a minimal chain decomposition $\mathcal{C} = \{C_1, \dots, C_w\}$ of $S$. We apply the coloring scheme defined above to $S$ based on $\mathcal{C}$, using exactly $w$ distinct colors.

Now, let $\mathcal{C}'$ be an arbitrary chain decomposition of $S$. According to the Lemma, a finite sequence of legal merges transforms $\mathcal{C}'$ into a decomposition $\mathcal{C}''$ where the minimal elements of all chains have pairwise distinct colors. Since there are only $w$ colors available in total, the number of chains in $\mathcal{C}''$ cannot exceed $w$. On the other hand, by Dilworth's theorem, no chain decomposition of $S$ can have fewer than $w$ chains. It follows that the final number of chains in $\mathcal{C}''$ is precisely $w$.
\end{proof}

\paragraph*{Acknowledgements}
I am grateful to Professor Douglas West for reading the manuscript and providing detailed comments.

\end{document}